\documentclass[twoside]{article}
\usepackage{amsmath}
\usepackage{amsthm}
\usepackage{amsfonts}
\usepackage{dsfont}
\usepackage[margin=1.35in]{geometry}
\usepackage{color}
\usepackage{bbm,mathtools,pdfsync}

\usepackage{datetime} \usdate
\usepackage{currfile}

\newtheorem{theorem}{Theorem}[section]

\newtheorem{corollary}[theorem]{Corollary}

\newtheorem{definition}[theorem]{Definition}

\newtheorem{lemma}[theorem]{Lemma}

\newtheorem{proposition}[theorem]{Proposition}

\theoremstyle{definition}  
\newtheorem{remark}[theorem]{Remark}

\numberwithin{equation}{section}

\newif\ifdraft
\ifdraft
\date{\today} 

\makeatletter
\@mparswitchfalse
\makeatother

\fi

\newcommand{\inv}{^{-1}}
\newcommand{\sgn}[1]{{\ssgn}(#1)}
\newcommand{\R}{\mathbb{R}}
\newcommand{\ds}{\,ds}
\newcommand{\dr}{\,dr}
\newcommand{\e}{\varepsilon}

\newcommand{\tN}{N}  
\newcommand{\bbone}{{\mathbbm 1}}

\DeclareMathOperator{\dom}{{\rm dom}}

\DeclareMathOperator{\ssgn}{{\rm sgn}}
\newcommand{\set}[1]{\{#1\}}

\mathtoolsset{showonlyrefs=true}

\title{Polynomial decay to equilibrium for the Becker-D\"{o}ring equations}
\author{Ryan W. Murray\thanks{Department of Mathematical Sciences and Center for Nonlinear Analysis, Carnegie Mellon University, Pittsburgh, PA, USA. Email: rwmurray@andrew.cmu.edu, rpego@cmu.edu .}\ \ and\ \ Robert L. Pego\footnotemark[1]}

\usepackage{fancyhdr}
\begin{document}
\maketitle

\let\oldlangle\langle
\let\oldrangle\rangle
\renewcommand{\langle}{\left\oldlangle}
\renewcommand{\rangle}{\right\oldrangle}

\begin{abstract}
This paper studies rates of decay to equilibrium for the  Becker-D\"{o}ring equations
with subcritical initial data. In particular, polynomial rates of decay are established when initial perturbations of equilibrium have polynomial moments. This is proved by using new dissipation estimates in polynomially weighted $\ell^1$ spaces, operator decomposition techniques from kinetic theory, and interpolation estimates from the study of travelling waves.
\end{abstract}

\noindent
{\bf{Keywords:}} Coagulation-fragmentation equations, rate of decay to equilibrium, interpolation inequalities.

\noindent
{\bf{AMS Mathematics Subject Classification:}} Primary: 34A35, 34D05. Secondary: 47D06, 82C05.

\section{Introduction}

In this work we consider the Becker-D\"oring equations, namely the following (infinite) system of differential equations
\begin{equation} \label{OriginalBDEquation}
\begin{aligned}
\frac{d}{dt} c_i(t) &= J_{i-1}(t) -J_{i}(t),\quad i=2,3,\ldots, \\
\frac{d}{dt}  c_1(t) &= -J_1(t) - \sum_{i=1}^\infty J_i(t),
\end{aligned}
\end{equation}
where the $J_i$ can be written as
\begin{equation} \label{OriginalFluxDefinition}
J_i(t) = a_ic_1(t) c_i(t) - b_{i+1} c_{i+1}(t),
\end{equation}
and where $(a_i),(b_i)$ are fixed, positive sequences, known as the coagulation and fragmentation coefficients respectively.

Becker-D\"oring systems form a subclass of the more general \emph{coagulation-fragmentation} equations. In typical physical applications the $c_i$ represent the discrete distribution function of particles of size $i$, and the evolution 
given by \eqref{OriginalBDEquation} represents the mean field approximation of the evolution of the distribution function $c_i$.  In particular, $J_i(t)$ represents the net rate 
that particles of size $i$ and size $1$ either join to form particles of size $i+1$,
or conversely are emitted by spontaneous breakup.
Thus we are primarily interested in positive solutions, whose first moment is preserved in time, meaning that
\begin{equation} \label{OriginalMassConstraint}
c_i \geq 0, \qquad \sum_{i=1}^\infty i c_i(t)  = \varrho(t) \equiv \varrho, \quad \forall t \geq 0.
\end{equation}
The Becker-D\"oring equations are used to model reactions in various physical settings, such as vapor condensation, phase separation in alloys and crystallization. This model was first proposed in \cite{BeckerDoring}, and was modified to the form we are considering in \cite{Burton},\cite{PenroseLebowitz}. A good mathematically-oriented review can be found in \cite{Slemrod}.

The well-posedness and convergence properties of the Becker-D\"oring equations have been well-studied. In particular, Ball, Carr and Penrose \cite{BallCarrPenrose}  demonstrated the existence of ``mass''-preserving, non-negative solutions to this system, namely solutions of \eqref{OriginalBDEquation} satisfying \eqref{OriginalMassConstraint}. A later work \cite{LaurencotMischlerUniqueness} established well-posedness (including uniqueness) for any initial data with finite first moment, namely the space where the ``mass'' is well-defined. Ball et al.\ \cite{BallCarrPenrose} also demonstrated that as $t \to \infty$ solutions must converge to some equilibrium $(Q_i)$, where $(Q_i)$ is uniquely determined by $\varrho$. Furthermore, they prove the existence of a value $\varrho_s$ such that if $\varrho < \varrho_s$ then the convergence to $(Q_i)$ is strong. On the other hand, if $\varrho > \varrho_s$ then there is a loss of mass to $\infty$, and the convergence is only weak. Any initial data satisfying $\varrho<\varrho_s$ is called \emph{subcritical}, while data satisfying $\varrho > \varrho_s$ is \emph{supercritical}.

The goal of this paper is to quantify the trend to equilibrium in the subcritical case ($\varrho < \varrho_s$). Specifically, our goal is to establish uniform, local rates of convergence to equilibrium in spaces with polynomial moments.

We define the detailed balance coefficients, a sequence $(\tilde Q_i)$, by the equations
\begin{equation}\label{DetailedBalance}
\tilde Q_1 = 1,\qquad
\tilde Q_i a_i = \tilde Q_{i+1}b_{i+1}, \quad  i=1,2,\ldots
\end{equation}
We note that the equilibrium solution $Q_i$ can be written as
\begin{equation}\label{QTildeDefinition}
Q_i = \tilde Q_i z^i,
\end{equation}
where the parameter $z$ is related to the mass $\varrho$ in the subcritical regime through the equation
\begin{equation}\label{zDetermined}
\sum_{i=1}^\infty i Q_i = \varrho.
\end{equation}
We note that $\varrho_s$ is linked to the radius of convergence $z_s$ of the power series with coefficients $\tilde Q_i$.

Part of our interest in studying these equations is precisely that we believe that the Becker-D\"oring equations are a suitable prototype of more general coagulation-fragmentation equations with detailed balance. Indeed, we suspect that many of the interesting phenomenon that occur for the Becker-D\"oring equations may be typical of other systems with detailed balance.

Convergence to equilibrium was proven by Ball, Carr and Penrose \cite{BallCarrPenrose} using an entropy functional. Specifically, they prove that the quantity
\[
\tilde V(c) := \sum_{i=1}^\infty c_i\left( \log\frac{c_i}{\tilde Q_i} - 1 \right)
\]
is weak-$*$ continuous and that $\tilde V(c(t))$ is strictly decreasing.

Later, Jabin and Niethammer \cite{JabinNiethammer} proved an entropy dissipation inequality which gives a uniform dissipation rate for regular data. In particular, they proved that if the initial data decays exponentially fast, then the solution converges to equilibrium with a rate bounded by $e^{-Ct^{1/3}}$ in the mass-weighted space.

In a recent work,  Ca\~{n}izo and Lods \cite{CanizoLods} improved this bound to $e^{-Ct}$. They do so by observing  that the Becker-D\"oring equations \eqref{OriginalBDEquation} have a type of symmetric structure, which we make use of below. In particular, if we write the Becker-D\"oring equations in terms of a perturbation of the equilibrium solution
\begin{equation}\label{hDefinition}
c_i = Q_i (1 + h_i),
\end{equation}
then we may express the original equation \eqref{OriginalBDEquation} in the form
\begin{equation}\label{BDEq}
\frac{d}{dt} h = F(h_1(t)) h,
\end{equation}
and the mass constraint \eqref{OriginalMassConstraint} as
\begin{equation} \label{hMassConstraint}
0 = \sum_{i=1}^\infty Q_i i h_i .
\end{equation}
We note that 
\begin{equation}\label{Def:F}
F(g) = L + g\Gamma,
\end{equation}
 where $L$ and $\Gamma$ are both linear operators. Ca\~{n}izo and Lods rewrote the operator $F(g)$ in weak form, satisfying
\begin{equation}
\sum_{i=1}^\infty Q_i(F (g)h)_i \phi_i  = \sum_{i=1}^\infty a_iQ_i Q_1 (h_1 + h_i - h_{i+1} + gh_i)(\phi_{i+1} - \phi_i - \phi_1)
\end{equation}
for all $(\phi_i)$ in a suitable space of test sequences.
They then linearized the equation and consider the operator $L$, which is given in weak form by
\begin{equation}\label{Def:L}
\sum_{i=1}^\infty Q_i(L h)_i \phi_i  = \sum_{i=1}^\infty a_iQ_i Q_1 (h_1 + h_i - h_{i+1} )(\phi_{i+1} - \phi_i - \phi_1).
\end{equation}
If we consider an $\ell^2$ space weighted by $Q_i$ then this form is clearly symmetric. Additionally, if $(c_i)$ is a solution of \eqref{OriginalBDEquation} and $(h_i)$ is determined by \eqref{hDefinition} then we have that $h_i \in [-1,\infty)$ and that $\sum Q_i i h_i = 0$. We thus define the Hilbert space $H$ by
\begin{equation}\label{Def:HSpace}
H := \left\{(h_i) : \|h\|_{\ell^2(Q_i)} := \left(\sum_{i=1}^\infty Q_i h_i^2 \right)^{1/2} < \infty,\quad \sum Q_i i h_i = 0\right\}.
\end{equation}
with the natural induced norm $\|\cdot \|_H = \|\cdot\|_{\ell^2(Q_i)} $ and inner product $\langle \cdot, \cdot \rangle_H$.
Ca\~{n}izo and Lods demonstrated that the linear part ($L$) of the Becker-D\"oring equations has a good spectral gap in $H$, or precisely that for some constant $\lambda_c > 0$ the following holds, independent of $h$:
\begin{equation} \label{Eqn:L2Gap}
\langle h,L h \rangle_H = -\sum_{i=1}^\infty a_iQ_i Q_1 (h_1 + h_i - h_{i+1})^2 \leq -\lambda_c \langle h,h \rangle_H.
\end{equation}
A key point is that the mass constraint \eqref{hMassConstraint} precludes the null vector $h_i = i$. Detailed quantitative estimates of $\lambda_c$ can then be obtained using Hardy's inequality---see \cite{CanizoLods} for details.

Ca\~{n}izo and Lods then utilized a priori bounds from \cite{JabinNiethammer} to control the non-linear term and establish a rate of convergence to equilibrium. More precisely, defining the Banach space
\[
X_\eta := \left\{ (h_i) : \|h\|_{\ell^1(Q_i e^{\eta i})} := \sum_{i=1}^\infty Q_i e^{\eta i} |h_i| < \infty,\quad  \sum Q_i i h_i = 0\right\},\quad 0 < \eta < 1,
\]
with the induced norm $\|\cdot \|_{X_\eta}=\|\cdot \|_{\ell^1(Q_i e^{\eta i})}$, they prove that for $0 < \eta < \bar \eta$, given initial data in $X_{\bar \eta}$ then the solution must converge at a uniform exponential rate in $X_\eta$. A key technical aspect of their proof was an operator decomposition technique from \cite{MouhotMischler}, which permits an extension of the spectral gap of $L$ from $H$ to $X_\eta$. We recall  (see  \cite{CanizoLods}) that the space $H$ is continuously embedded in $X_\eta$ for $\eta>0$ sufficiently small, precisely because the $Q_i$ are exponentially decaying.

Our aim in the present paper is to analyze the trend to equilibrium for a wider class of initial data, for which the a priori bounds from \cite{JabinNiethammer} are not available. We define the Banach spaces
\begin{equation}
X_k := \left\{ (h_i) : \|h\|_{\ell^1(Q_i i^k)} := \sum_{i=1}^\infty Q_i i^k |h_i| < \infty,\quad  \sum Q_i i h_i = 0\right\},\quad k\geq 1,\\
\end{equation}
with norm $\|\cdot\|_{X_k}=\|\cdot\|_{\ell^1(Q_i i^k)}$. The main result of our paper is as follows.

\begin{theorem} \label{Thm:AlgebraicDecay}
Let $(h_i(t))$  defined by \eqref{hDefinition} represent the deviation from equilibrium of a solution $(c_i(t))$ to the Becker-D\"oring equations (see Definition \ref{Def:Solution}). Assume that the model coefficients in \eqref{OriginalFluxDefinition} satisfy \eqref{aLowerBound}-\eqref{boundedByI} below. Let  $m$ and $k$ be real numbers satisfying $m> 0$ and $k>m+2$. Then there exists positive constants $\delta_{k,m}, C_{k,m}$ so that if $\|h(0)\|_{X_{1+k}} < \delta_{k,m}$ then we have that
\begin{equation}
\|h(t)\|_{X_{1+m}} \leq C_{k,m}(1+t)^{-(k-m-1)}\|h(0)\|_{X_{1+k}} \quad \mbox{for all $t \geq 0$}.
\end{equation}
\end{theorem}

In order to obtain this result, we establish detailed estimates on the semigroup generated by $L$ in the spaces $X_k$ by using new dissipation estimates, together with the spectral gap estimate \eqref{Eqn:L2Gap}, the operator decomposition result from \cite{MouhotMischler} and interpolation techniques from Engler's work on travelling wave stability \cite{Engler}. We then prove a local stability result in $X_k$ (see Theorem \ref{Lem:LyapunovStability}), which along with Duhamel's formula proves the desired result.



\subsection{Assumptions and Preliminaries}\label{Sec:Assumptions}

We impose the following assumptions on our model coefficients:
\begin{align}
a_i &>C_1>0\qquad \forall i \geq 1, \label{aLowerBound}\\
\lim_{i \to \infty} \frac{a_{i+1}}{a_i} &= 1, \label{aLimit}\\
\lim_{i \to \infty} \frac{a_i}{b_{i}} &=: \frac{1}{z_s}\in(0,\infty) \label{tildeQLimit}\\
 a_i, b_i &\leq C_2 i   \qquad \forall i\geq 1, \label{boundedByI}\\
\end{align}
with $a_i,b_i$ as in \eqref{OriginalBDEquation} and \eqref{OriginalFluxDefinition}, and where $C_1,C_2$ are fixed constants, independent of $i$.

Following \cite{BallCarrPenrose}, we define a solution to the Becker-D\"oring equations in the following way
\begin{definition}\label{Def:Solution}
    A function $(c_i(t))$ is a solution to the Becker-D\"oring equations on $[0,T)$ if
\begin{enumerate}
\item $\sum_{i=1}^\infty i |c_i| < \infty$ for all $t \in [0,T)$.
\item For all $i$ we have that  $c_i(t)$ is continuous in time, and non-negative.
\item The following (well-defined) equations are satisfied
\end{enumerate}
\begin{align}
c_i(t) &= c_i(0) + \int_0^t \left( J_{i-1}(s) - J_i(s)\right) \ds,\qquad i \geq 2,\\
c_1(t) &= c_1(0) - \int_0^t \left(J_1(s) + \sum_{i=1}^\infty J_i(s)\right) \ds.
\end{align}
\end{definition}

Throughout the paper we will be considering solutions $(c_i(t))$ of the Becker-D\"oring equations \eqref{OriginalBDEquation} with some fixed, subcritical mass, meaning that for some $z < z_s$, we have that the $Q_i$ defined by \eqref{QTildeDefinition} will satisfy
\[
\sum_{i=1}^\infty Q_i i = \varrho = \sum_{i=1}^\infty i c_i(t).
\]
Using  \eqref{DetailedBalance}, \eqref{QTildeDefinition} and \eqref{tildeQLimit}, it is immediate that
\begin{equation}
\lim_{i \to \infty} \frac{Q_{i+1}}{Q_i} = \frac{z}{z_s}  < 1. \label{QLimit}
\end{equation}
This naturally implies that the $Q_i$ are exponentially decaying.

Also, by combining \eqref{aLimit} and \eqref{tildeQLimit}, we observe that
\begin{equation}
a_i(z + \delta)  = a_i(Q_1 + \delta) \leq b_i, \quad \forall i > N_z , \label{asymptoticallyStrongFragmentation}
\end{equation}
for some $\delta>0$ and $N_z$ that are fixed and independent of $i$, but possibly dependent on $z$. All of these assumptions are fairly standard, and versions of them can be found in \cite{BallCarrConvergence,CanizoLods,JabinNiethammer}. In particular we note the similarity of \eqref{asymptoticallyStrongFragmentation} with the assumption given in \cite{BallCarrConvergence}. In that work Ball and Carr make the assumption that
\begin{equation}
a_i z \leq b_i
\end{equation}
for $i>\hat N$, and for all $z<z_s$. In that work, this assumption was made in order to guarantee that $V(c(t_n))$ converges to the minimum value of $V$, where $V$ is a suitable entropy functional. In their work, coefficients were required to be $O(i/\log(i))$, but this was subsequently relaxed in \cite{SlemrodEquilibrium}. These assumptions were also utilized in \cite{CanizoLods} and \cite{JabinNiethammer}.

%
%
%
%


One of the primary advantages to our method is that it lays bare the mechanism causing convergence to equilibrium. Inequality \eqref{asymptoticallyStrongFragmentation} arises naturally in attempting to establish dissipation estimates, thus motivating the analytical need for such assumptions. More importantly, \eqref{asymptoticallyStrongFragmentation} is satisfied by many of the relevant physical models. For example, one physically-motivated form of the model coefficients is (see \cite{Penrose89})
\begin{equation}\label{PenroseCoefficients}
a_i = i^\alpha,\quad b_i = a_i\left(z_s + \frac{q}{i^{1-\mu}}\right),\quad \alpha \in (0,1],\quad \mu \in [0,1],\quad q > 0.
\end{equation}
For this model we have
\begin{equation}
b_i - Q_1 a_i  \geq (z_s-z)a_i,
\end{equation}
which naturally implies that assumption \eqref{asymptoticallyStrongFragmentation} is only satisfied in the subcritical setting.

\section{Linearized stability estimates In $X_1$}\label{Sec:X1Linear}

In this section we establish stability estimates for the semigroup generated by the operator $L$, in the space $X_1$.  Following \cite{CanizoLods}, our goal is to use an operator decomposition technique to derive uniform bounds on $e^{Lt}$ in $X_1$. This technique was first developed by Gualdani, Mischler and Mouhot \cite{MouhotMischler} to study the Boltzmann equation, and was previously applied to the Becker-D\"oring equations by Canizo and Lods \cite{CanizoLods}. Here we generalize the technique to the case of evolution families.

  We remark that the symbols $M$ and $C$, with various subscripts, will represent generic constants whose value may sometimes vary line by line.
The notation ${\mathcal L}(Y,Z)$ denotes the space of bounded linear operators from $Y$
to $Z$, and ${\mathcal L}(Y)={\mathcal L}(Y,Y)$.
In this paper, the term ``semigroup'' always refers to a strongly continuous semigroup
of linear operators.

\begin{proposition}[Extension Principle] \label{PropMischlerMouhot}
Let $Z \subset Y$ be Banach spaces, with $Z$ continuously embedded into $Y$. Let $I=[0,T)$ with $T = \infty$ permitted, and let $\{A(t)\}_{t \in I}$ and $\{B(t)\}_{t \in I}$ be families of linear operators on $Y$. Suppose that
\begin{enumerate}
    \item $\set{A(t)+B(t)}_{t\in I}$ generates an evolution family $U^Z$ on $Z$, satisfying 
\[
    \|U^Z(t,s)\|_{\mathcal{L}(Z)} \leq M_Ze^{-\lambda_Z (t-s)} 
    \qquad\text{ for } 0\le s\le t<T, 
\]a
for some $\lambda_Z \in\R$.
\item $B(t)$ is ``regularizing,'' meaning that $B(\cdot) \in C(I;\mathcal{L}(Y,Z))$, and that $\|B(t)\|_{\mathcal{L}(Y,Z)}< M_B$, uniformly for $t\in I$.
\item $\set{A(t)}_{t\in I}$ generates an evolution family $V$ on $Y$, satisfying 
\[
\|V(t,s)\|_{\mathcal{L}(Y)} \leq M_V e^{-\lambda_Y (t-s)}
    \qquad\text{ for } 0\le s\le t<T, 
\]
with $\lambda_Z > \lambda_Y $.
\end{enumerate}
Then $\set{A(t)+B(t)}_{t\in I}$ generates an evolution family $U^Y$ on $Y$ with bound
\begin{equation} \label{Eqn:MMBound}
\|U^Y(t,s)\|_{\mathcal{L}(Y)} \leq M_Ye^{-\lambda_Y (t-s)}
    \qquad\text{ for } 0\le s\le t<T.
\end{equation}
\end{proposition}

\begin{proof}
This proof is almost identical to that of Theorem 3.1 in \cite{CanizoLods}, with the necessary changes to the setting of evolution families. We give the proof for clarity.

First, it is evident that $\set{A(t) + B(t)}_{t\in I}$ generates an evolution family since $B(t)$ is bounded and continuous in $t$ (see, for example, Theorem 5.2.3 in \cite{PazyBook}). Thus our goal is to prove \eqref{Eqn:MMBound}.

Using Duhamel's formula we can write the evolution family generated by $A(t)+B(t)$ as follows:
\begin{equation}
U^Y(t,s) h(s) = V(t,s) h(s) + \int_s^t U^Y(t,r) (B(r)V(r,s) h(s)) \dr
\end{equation}

This formula can be rigorously justified in the current setting by applying Lemma 5.4.5 in \cite{PazyBook}. We then estimate
\begin{equation}
\|U^Y(t,s) h(s)\|_Y \leq M_Ve^{-\lambda_Y (t-s)}\|h(s)\|_Y + 
\int_s^t \|U^Y(t,r)B(r)V(r,s) h(s)\|_Y \dr.
\end{equation}


As $B$ maps from $Y$ to $Z$ we can replace $U^Y$ with $U^Z$ inside the integral, and then estimate using the decay estimate in $Z$ to infer
\begin{equation}
\|U(t,s)^Y h(s)\|_Y \leq M_Ve^{-\lambda_Y(t-s)}\|h(s)\|_Y + \int_s^t M_Z e^{-\lambda_Z(t-r)} \|B(r)V(r,s) h(s)\|_Z \dr.
\end{equation}
Using our bounds on $B$ and $V$ we obtain
\begin{align}
\|U^Y(t,s) h(s)\|_Y &\leq M_Ve^{-\lambda_Y(t-s)} \|h(s)\|_Y  + \|h(s)\|_Y M_VM_ZM_Be^{-\lambda_Y(t-s)} \int_s^t e^{-(\lambda_Z - \lambda_Y)(t-r)} \dr \\
&\leq M_Y e^{-\lambda_Y (t-s)} \|h(s)\|_Y,
\end{align}
which is the desired result.
\end{proof}

\begin{remark}
When $A$ and $B$ are constant in time this reduces to a statement about semigroups, and indeed in that case the statement and proof are found in \cite{CanizoLods}. In this section we only use the proposition to prove bounds on the semigroup $e^{Lt}$, but in Section \ref{Sec:NonLinear} we will use it in the case of evolution families.
\end{remark}

We emphasize that the previous result is valid when $\lambda_Y = 0$, meaning that the result applies to semigroups that are only stable.

Next, recall that the operator $L$ is determined by the weak form \eqref{Def:L}. 
We write
\begin{equation}
L=A+B \,,
\end{equation}
with the operator $A$ determined 
via the weak form
\begin{equation} \label{def:AOperator}
\sum_{i=1}^\infty Q_i (Ah)_i \phi_i := \sum_{i = \tN}^\infty Q_i Q_1 a_i (h_i - h_{i+1})(\phi_{i+1}-\phi_i-\phi_1) - Q_{\tN-1} Q_1a_{\tN-1} h_{\tN}(\phi_{\tN} - \phi_{\tN-1} - \phi_1),
\end{equation}
where we fix some $\tN \geq N_z +1$, with $N_z$ given in \eqref{asymptoticallyStrongFragmentation}. We take the domain of definition for both $A$ and $L$ initially to be the set of sequences with finite support that satisfy \eqref{hMassConstraint}, namely having zero ``mass''. We note that if we set $\phi_i = i$ we get zero, implying that $A$ and $B$ both map into the space of sequences with zero mass.

We first give an elementary bound on $L$ and $\Gamma$, which indicates a minimal size for the domain of the closure of these operators. We will subsequently show that $B$ is bounded, which in turn means that this also gives information about the domain of the closure of $A$.

\begin{lemma}\label{Lemma:GammaEstimate}
For any $m\geq 0$, and for some constant $C_m$ the following bound holds
\[
\|\Gamma h\|_{X_{1+m}} \leq C_m \|h\|_{X_{2+m}}\qquad \|Lh\|_{X_{1+m}} \leq C_m \|h\|_{X_{2+m}}.
\]
\end{lemma}

\begin{proof}
We only show the estimate for $L$, as the estimate for $\Gamma$ is essentially identical. We simply estimate
\begin{align}
\|Lh\|_{X_{1+m}} &= \sum_{i=0}^\infty Q_i (L h)_i i^{1+m} \ssgn(L h)_i \\
&\leq \sum_{i=1}^\infty Q_i(a_i Q_1 + b_i) |h_i| 3(i+1)^{1+m}  + 3|h_1|\sum_{i=1}^\infty Q_iQ_1 a_i (i+1)^{1+m}\\
&\leq C \sum_{i=1}^\infty Q_i i^{2+m}  |h_i|,
\end{align}
where we have used \eqref{boundedByI}. This proves the lemma.
\end{proof}

In order to use the extension principle, Proposition \ref{PropMischlerMouhot}, we first prove that $B$ is ``regularizing.''
(Recall $H\subset X_1$.) 

\begin{lemma}\label{Lemma:BRegularizing}
The operator $B$ is a bounded operator from $X_1$ to $H$.
\end{lemma}

\begin{proof}
We compute in weak form:
\begin{align}
\sum_{i=1}^\infty Q_i (Bh)_i \phi_i &= \sum_{i=1}^{\tN-2} Q_iQ_1 a_i ( h_i-h_{i+1}) (\phi_{i+1} - \phi_i - \phi_1) + \sum_{i=1}^\infty Q_i Q_1 a_i h_1 (\phi_{i+1} - \phi_i - \phi_1)\\
&+ Q_{\tN - 1}Q_1a_{\tN - 1} h_{\tN - 1}(\phi_{\tN} - \phi_{\tN - 1} - \phi_1)  \\
 &=: B_1(h,\phi) + B_2(h,\phi) + B_3(h,\phi).
\end{align}
By Cauchy-Schwarz, the fact that $0<c\leq {Q_i}/{Q_{i+1}} \leq C < \infty$ by \eqref{QLimit}, and the equivalence of finite dimensional norms,
\[
|B_1(h,\phi)| \leq C \left(\sum_{i=1}^{\tN-1} Q_i \phi_i^2\right)^{1/2} \left(\sum_{i=1}^{\tN-1} Q_i h_i^2\right)^{1/2} \leq C \|\phi\|_H \|h\|_{X_1}.
\]
Furthermore,
\[
|B_2(h,\phi)| \leq C|h_1|\left(\sum_{i=1}^\infty Q_i a_i^2 \right)^{1/2} \left( \sum_{t=1}^\infty Q_i \phi_i^2 \right)^{1/2} \leq C\|h\|_{X_1}\|\phi\|_H.
\]
Similarly, $|B_3(h,\phi)| \leq C\|h\|_{X_1}\|\phi\|_H$. Setting $\phi = Bh$ then establishes the desired result.
\end{proof}

Next we need to show that $A$, or more precisely its closure, generates a contraction semigroup on $X_1$. We recall the following definition from Pazy \cite{PazyBook}.

\begin{definition}\label{Def:Dissipation}
Let $x \in X$, with $X$ a Banach space. Define
\begin{equation}\label{Def:JSet}
\mathcal{J}(x) := \left\{x^* \in X^* : \langle x^*,x\rangle_{X^*,X} = \|x\|_X^2 = \|x^*\|_{X^*}^2 \right\}.
\end{equation}
A linear operator $A$ with domain of definition $\dom(A) \subset X$ is called \emph{dissipative} if for every $x \in \dom(A)$ there exists an $x^* \in \mathcal{J}(x)$ such that
\[
\langle x^*, Ax\rangle_{X^*,X} \leq 0
\]
\end{definition}

By way of notation, when $X=\ell_1(Q_iw_i)$ and $\|h\|_X = \sum_{i=1}^\infty Q_i w_i |h_i|$ we will write 
\begin{equation}\label{Def:SgnAsDual}
\langle \sgn{h},\phi \rangle_{X^*,X} := \sum_{i=1}^\infty Q_i w_i \phi_i \sgn{h_i} \,.
\end{equation}
By the definition of $\mathcal{J}(x)$, namely \eqref{Def:JSet}, it is clear that if $\langle \sgn h,Ah \rangle_{X^*,X} \leq 0$ for all $h$ in the domain of definition of $A$ then $A$ is dissipative.

\begin{proposition} \label{Prop:ADissipative}
The operator $A$ given by \eqref{def:AOperator} is dissipative on $X_1$.
\end{proposition}

\begin{proof}
Rearranging our sum and using \eqref{DetailedBalance} to say $Q_iQ_1a_i=Q_{i+1}b_{i+1}$,  we find that
\begin{align}
&\langle \sgn{h} , A h \rangle_{X_1^*,X_1}  
\\& \quad = 
\sum_{i = \tN}^\infty Q_i Q_1a_i h_i ((i+1)\sgn{h_{i+1}} - i \,\sgn{h_i} - \sgn{h_1})) 
\\&\qquad -
\sum_{i = \tN}^\infty Q_i b_i h_i (i\, \sgn{h_i} - (i-1)\sgn{h_{i-1}} - \sgn{h_1}) 
\\&\quad=
\sum_{i = \tN}^\infty  Q_i h_i \Bigl(
Q_1a_i(i+1)(\sgn{h_{i+1}} -\sgn{h_i})  
+ b_i(i-1)(\sgn{h_{i-1}}-\sgn{h_i} ) 
\Bigr)
\\&\qquad+ 
\sum_{i = \tN}^\infty Q_i|h_i| (a_iQ_1 - b_i)
+ 
\sgn{h_1} \sum_{i = \tN}^\infty Q_ih_i(b_i - Q_1a_i)
\\&\quad=: E_1 + E_2 + E_3,
\end{align}
Because $h_i(\sgn{h_{i\pm1}}-\sgn{h_i})\le0$, we see $E_1 \leq 0$. 
Furthermore, we have that
\[
    E_2 + E_3 = 2 \sum_{\mathclap{\substack{i=\tN \\ \sgn{h_1} \neq \sgn{h_i}}}}^\infty 
Q_i|h_i|(a_iQ_1 - b_i)\leq 0 \,.
\]
This readily implies that $A$ is dissipative (see Definition \ref{Def:Dissipation}).
\end{proof}

\begin{remark}
In the case that $a_i \sim i$ it is probably possible to prove that $L$ has a spectral gap in $X_1$. We do not pursue this line of analysis, because in most of the physical cases $a_i \sim i^\alpha$, with $\alpha \in (0,1)$.
\end{remark}

Next we recall two results from \cite{CanizoLods} (Corollary 2.11 and Theorem 3.5),
that concern the closure of $L$ (which we also denote below by $L$).

\begin{proposition}\label{Prop:CanizoLSemigroup}
For some $\lambda_c>0$, the operator $L$ generates a contraction semigroup $e^{Lt}$ on $H$ satisfying
\[
\|e^{Lt}\|_{\mathcal{L}(H)} \leq e^{-\lambda_c t} 
\quad\text{ for all $t\ge0$.}
\]
Furthermore, for $\eta>0$ sufficiently small there exists constants $M$ and $\lambda_\eta>0$ so that the operator $L$ generates a semigroup on $X_\eta$ satisfying
\[
\|e^{Lt}\|_{\mathcal{L}(X_\eta)} \leq Me^{-\lambda_\eta t}
\quad\text{ for all $t\ge0$.}
\]
\end{proposition}

Next we prove that the closure of $A$ indeed generates a semigroup. 

\begin{lemma}\label{Lem:AContraction}
The closure of $A$ (which we also denote by $A$), generates a contraction semigroup on $X_1$.
\end{lemma}

\begin{proof}
The Lumer-Phillips theorem (see e.g. \cite{EngelNagel} Theorem II.3.15) states that the following are equivalent for a densely-defined, dissipative operator $A$:
\begin{enumerate}
\item The range of $(A-\lambda I)$ is dense for some $\lambda >0$.
\item $A$ is closable and its closure (also denoted by $A$) generates a contraction semigroup.
\end{enumerate}

We know that $H \subset X_1$, and that $H$ is dense in $X_1$. By Proposition \ref{Prop:CanizoLSemigroup} we know that $L$ generates a contraction semigroup on $H$. As $B$ is bounded on $H$, we know that $A$ (restricted to $H$) generates a semigroup on $H$. Thus it must be (see e.g. Theorem 1.5.3 in \cite{PazyBook}) that for $\lambda >0$ large enough $A-\lambda I$ is invertible on $H$. Thus the range of $A-\lambda I$  contains $H$, and thus is dense in $X_1$. Because $A$ is dissipative by Proposition \ref{Prop:ADissipative}, the Lumer-Phillips theorem then implies that $A$ generates a contraction semigroup on $X_1$.
\end{proof}

By combining Proposition \ref{PropMischlerMouhot} and Lemmas \ref{Lemma:BRegularizing} and \ref{Lem:AContraction} along with Proposition \ref{Prop:CanizoLSemigroup} we immediately obtain the following:

\begin{theorem} \label{Cor:LStableSemigroup}
The closure of $L$ generates a semigroup $e^{Lt}$ on $X_1$ uniformly bounded in time:
\[
\|e^{Lt}\|_{\mathcal{L}(X_1)} \leq M \qquad\text{for all $t\ge0$}.
\]
\end{theorem}

\section{Polynomial Decay Estimates}\label{Sec:PolyDecay}

In this section we prove Theorem \ref{Thm:AlgebraicDecay}, namely that perturbations of equilibrium small in $X_k$ will decay with a uniform, polynomial rate. We will first prove polynomial decay results for $e^{Lt}$. The following interpolation result is a modification of a theorem in \cite{Engler}, where it was originally used to study the convergence of travelling waves.

\begin{theorem}\label{SemigroupInterpolation}
Let $\eta\in(0,1)$ and $m,k\in\R$ with $0<m<k$. Let $\set{S(t)}_{t\ge0}$ be a family of linear operators on $X_1$
which for any $t>0$ satisfies
\[
\|S(t)u - S(t) v\|_{X_1} \leq M\|u-v\|_{X_1}, \qquad \|S(t)u\|_{X_\eta} \leq Me^{-\lambda_\eta t}\|u\|_{X_\eta},
\]
where $u,v$ are arbitrary elements of the appropriate spaces, $M$ is a fixed positive constant and $\lambda_\eta > 0$. 
Then the operators $S(t)$ necessarily are bounded from $X_{1+k}$ to $X_{1+m}$
and satisfy 
\begin{equation}
    \|S(t) u\|_{X_{1+m}} \leq C(1+t)^{-(k-m)} \|u\|_{X_{1+k}}
    \quad\text{for all $u\in X_{1+k}$ and $t\ge0$,}
\end{equation}
where $C$ depends on $m,k,M$ and $\lambda_\eta$.
\end{theorem}

\begin{proof}
The proof is very similar to that found in \cite{Engler}, with modifications necessary, however,  to handle the mass constraint and weighted norm on $X_1$. 

1. Consider $K: \mathbb{R} \times X_1 \to \mathbb{R}$ defined by
\begin{equation}
K(s,u) = \inf_{v \in X_\eta}(\|u-v\|_{X_1} + e^{s}\|v\|_{X_\eta}).
\end{equation}
In interpolation theory \cite{BerghLofstrom} this is known as a modified K-functional. 
For fixed $s$, $K(s,\cdot)$ is a norm.
Clearly $K(s,u)$ is increasing in $s$ and bounded above by $\|u\|_{X_1}$. 
Furthermore, we claim that $K$ is absolutely continuous in $s$. Indeed, if we define $\tilde K(\tilde s,u) := K(\log \tilde s, u)$, then $\tilde K(\cdot,u)$ can be written an the infimum of affine functions, and thus must be concave. This readily implies that $K(s,u)$ is absolutely continuous in $s$.

We begin by proving upper and lower bounds on $K$. First, we get the lower bound
\begin{align}\label{Eqn:KLower}
K(s,u) &\geq 
\sum_{i=1}^\infty Q_i \inf_{v\in\R} (|u_i-v|i + e^{s + \eta i} |v|) 
= \sum_{i=1}^\infty Q_i |u_i|(i \wedge e^{s + \eta i}) \,.
\end{align}
Next, observe $x\wedge e^{s+\eta x}=x$ for all real $x$ if and only if 
$s\ge s_\eta:=-1-\log \eta$. Thus for $s\ge s_\eta$,
\begin{equation}
K(s,u) \le \|u\|_{X_1} = \sum_{i=1}^\infty Q_i|u_i|(i\wedge e^{s+\eta i}) \,.
\end{equation}
Suppose now that $s<s_\eta$. Then 
$1/\eta\in \{x: e^{s+\eta x}\le x\} = [z_-,z_+]\subset(0,\infty)$.
Let $j(s)$ be the least integer greater than or equal to $z_+$,
and define the sequence $v_s(u)$ by
\begin{equation}
v_s(u)_i := \begin{cases} u_i &\text{ for } i < j(s), \\
({Q_{i} i})^{-1}
{\sum\limits_{k \geq j(s)} Q_kk u_k} &\text{ for } i = j(s) , \\
0 &\text{ for } i>j(s). 
\end{cases}
\end{equation}
In particular note that $\sum_{i=1}^\infty Q_i i v_s(u)_i=0$, so $v_s(u)\in X_\eta$.
Writing $j=j(s)$,
we then find 
\begin{align}
K(s,u) &\leq \|u-v_s(u)\|_{X_1} + e^s \|v_s(u)\|_{X_\eta} \\
&=\left| \sum_{i>j} Q_i i u_i \right|  +\sum_{i>j} Q_i i|u_i|   
+ e^s\sum_{i=1}^{j-1} Q_i e^{\eta i} |u_i|
+ e^s \frac{Q_{j} e^{\eta j}}{Q_{j} j} \left| \sum_{i=j}^\infty Q_i i u_i \right| \\
&\leq \left( 2 + \frac{e^{s + \eta j}}{j}\right) \sum_{i=j}^\infty Q_i i |u_i| + \sum_{i=1}^{j-1} Q_i e^{s + \eta i} |u_i| \,.
\end{align}
Now, $j\inv e^{s+\eta j}\le z_+\inv e^{s+\eta(z_++1)}=e^\eta$, and $i\ge j$
implies $i=i\wedge e^{s+\eta i}$.  Furthermore, whenever $1\le i\le z_-$ 
we have $e^{s+\eta i}\le e^{s+\eta z_-}=z_-\le 1/\eta \le i/\eta = (i\wedge e^{s+\eta i})/\eta$.
By these estimates we find that with $C=\max\{2+e^\eta,1/\eta\}$ we have
that for any $s\in\R$,
\begin{equation}
K(s,u) \le C \sum_{i=1}^\infty Q_i |u_i|(i\wedge e^{s+\eta i}) \,.
\label{Eqn:KUpper}
\end{equation}

2. In the next step, for $r > 0$ we set
\begin{equation}
h_r(s) := \begin{cases} e^{-s} &\text{ for } s \geq 0, \\
(1-s)^{r-1} &\text{ for } s \leq 0,
\end{cases}
\end{equation}
and define the norm
\begin{equation}
\|u\|_* := \int_\R K(s,u) h_r(s) \ds\,.
\end{equation}
We claim this is equivalent to the norm in $X_{1+r}$. 
By \eqref{Eqn:KLower} and \eqref{Eqn:KUpper}, it suffices to show there exist 
$C_-,C_+>0$ 
independent of $i$ such that
\begin{equation}
C_-(1+i)^{1+r} \le  \int_\R (i\wedge e^{s+\eta i})h_r(s)\,ds
\le
C_+(1+i)^{1+r} \quad\mbox{ for $i\ge 1$}.
\label{Eqn:hrbounds}
\end{equation}
To show this, we first bound the part of the integral over $s\in[0,\infty)$, finding that
\begin{equation} \label{Eqn:splusbound}
1\le \int_0^\infty (i\wedge e^{s+\eta i}) e^{-s}\,ds \le i
\le (1+i)^{1+r}\,.
\end{equation}
For the part over $s\in(-\infty,0]$, after changing variables twice via
$z=-s$, $\sigma=z-\eta i$,
we have 
\begin{align}
\int_{-\infty}^0 (i\wedge e^{s+\eta i}) (1-s)^{r-1}\,ds 
&\le i \int_0^\infty (1\wedge e^{-z+\eta i})(1+z)^{r-1}\,dz 
\\& =i\int_0^{\eta i} (1+z)^{r-1}\,dz 
+i \int_0^\infty e^{-\sigma}(1+\eta i+\sigma)^{r-1}\,d\sigma
\\& \le C i (1+\eta i)^r \le C(1+i)^{1+r} \,.
\end{align}
This establishes the upper bound in \eqref{Eqn:hrbounds}.

To get the lower bound, choose $I_\eta$ so large that $i>I_\eta$ implies 
$\eta i-\log i \ge \frac12\eta i$.  
For $i\le I_\eta$ we have $(1+i)^{r+1}\le (1+I_\eta)^{r+1}$,
hence we get the lower bound in \eqref{Eqn:hrbounds}
with $C_-=(1+I_\eta)^{-1-r}$ by using \eqref{Eqn:splusbound}.
For $i>I_\eta$, we find
\begin{align}
\int_{-\infty}^0 (i\wedge e^{s+\eta i}) (1-s)^{r-1}\,ds 
&=  i \int_0^\infty (1\wedge e^{-z+\eta i-\log i}) (1+z)^{r-1}\,dz
\\& \ge i\int_0^{\eta i/2} z^{r-1}\,dz \ge C(1+i)^{1+r}\,.
\end{align}
Thus $\|\cdot\|_*$ is equivalent to $\|\cdot\|_{X_{1+r}}$. 

3. Now, let $H_r(t) := \int_t^\infty h_r(\tau) d\tau$ and $k(t) := \frac{d}{dt} K(t,u)$. 
We claim that
\begin{equation}
H_m(s+t) \leq C H_k(s) (1+t)^{m-k},
\end{equation}
for all $s \in \R$, and for $t \geq 0$. To prove the claim, we first note that
\[
H_m(s) = \begin{cases} e^{-s} &\text{for } s \geq 0 ,\\ 1+ \frac{(1-s)^m - 1}{m} &\text{ for } s < 0 ,\end{cases}
\]
and furthermore, for $s<0$, we can find $c,C>0$ so that
\begin{equation}\label{Eqn:HkEquivalence}
c(1-s)^m \leq H_m(s) \leq C (1-s)^m.
\end{equation}
We then consider separate cases. First, if $s \geq 0$,
\[
H_m(s+t) = e^{-(s+t)} \leq C e^{-s} (1+t)^{m-k} =  C H_k(s) (1+t)^{m-k}.
\]
Next suppose that $s < 0 \leq s+t$. Then
\begin{align}
H_m(s+t) = e^{-(s+t)} &\leq C(1+s+t)^{-k} = C\frac{(1-s)^k}{(1+t-s(s+t))^k}
\le C(1+t)^{-k}H_k(s) \,,
\end{align}
where we have used \eqref{Eqn:HkEquivalence}. 
Finally, in the case that $t <-s$, we note that because $m-k<0$,
\[
(1-(s+t))^m \le (1-s)^m \le (1-s)^k(1+t)^{m-k} \,.
\]
In light of \eqref{Eqn:HkEquivalence} this proves the claim. 

4. Next, we use the assumed bounds on our operators to estimate
\begin{align}
K(s,S(t)u) &\leq \inf_{v \in X_\eta}(\|S(t)u - S(t)v\|_{X_1} + e^{s}\|S(t)v\|_{X_\eta})\\
&\leq M\inf_{v \in X_\eta}(\|u-v\|_{X_1} + e^{s-\lambda_\eta t} \|v\|_{X_\eta}) \\
&= MK(s-\lambda_\eta t,u).
\end{align}
We remark that for $u \in X_\eta$ we have that $0 \leq K(s,u) \leq 
\|u\|_{X_1}\wedge e^s\|u\|_{X_\eta}$, 
and thus for $u \in X_\eta$ we have that $H_r(s)K(s,u)$ goes to zero as $s \to \pm \infty$. Thus we may use integration by parts, and our previous estimates, to obtain the following for any $u \in X_\eta$:
\begin{align}
\|S(t)u\|_{X_{1+m}} &\leq C \int_\R K(s,S(t)u)h_m(s) \ds \\
& \leq C\int_\R K(s-\lambda_\eta t,u)h_m(s) \ds \\
& = C \int_\R k(s,u)H_m(s+\lambda_\eta t) \ds\\
&\leq C(1+t)^{m-k} \int k(s,u) H_k(s) \ds \\
&= C(1+t)^{m-k} \int_\R K(s,u) h_k(s) \ds \\
&= C(1+t)^{m-k} \|u\|_* \leq C(1+t)^{m-k} \|u\|_{X_{1+k}}.
\end{align}
Because  $X_\eta$ is dense in $X_{1+k}$, we have the desired inequality. This completes the proof.
\end{proof}

We will apply the previous theorem to the semigroup generated by $L$. We first state a proposition, which will be an immediate consequence of Lemma \ref{lem:AgEvolutionFamily}
in the following section.

\begin{proposition}\label{Cor:XkSemigroup}
The operator $L$ generates a  semigroup on the space $X_{1+k}$, for any $k\geq 0$.
\end{proposition}

\begin{proof}
Applying Lemma \ref{lem:AgEvolutionFamily} when $g \equiv 0$, along with the fact that $B$ is a bounded perturbation gives the desired result.
\end{proof}

With these tool in hand we can establish the following linear decay estimates.

\begin{corollary} \label{Cor:SemigroupPoly}
Provided $0 < m<k$, the semigroup $e^{Lt}$ generated by the operator $L$ satisfies
\[
\|e^{Lt} u\|_{X_{1+m}} \leq C(1+t)^{-(k-m)} \|u\|_{X_{1+k}}
    \qquad\text{for all $u\in X_{1+k}$,}
\]
where $C$ depends on $m$ and $k$, but not on $u$ or $t$.
\end{corollary}

\begin{proof}
This follows directly from Proposition \ref{Prop:CanizoLSemigroup}, Corollary \ref{Cor:LStableSemigroup}, Theorem \ref{SemigroupInterpolation} and Proposition \ref{Cor:XkSemigroup}.
\end{proof}

Our goal is to use the detailed decay rates in Corollary \ref{Cor:SemigroupPoly}, along with Duhamel's formula, to prove Theorem \ref{Thm:AlgebraicDecay}. We first prove that Duhamel's formula is justified in the appropriate spaces. We begin by recalling a fact from Ball, Carr and Penrose (\cite{BallCarrPenrose}, Proof of Theorem 2.2 and Proposition 2.4).

\begin{proposition}\label{Prop:BCP}
Let $(c_i)$ be a solution to the Becker-D\"oring equations, and let $(h_i)$ be defined by \eqref{hDefinition}. Suppose that $h(0) \in X_{1+k}$, with $k \geq 0$. Then $\|h(t)\|_{X_{1+k}} \leq \|h(0)\|_{X_{1+k}}Ce^{Kt}$ for some $C$ and $K$ independent of $h$.
\end{proposition}

We now have the tools to justify Duhamel's formula.

\begin{lemma} \label{Lem:Duhamel}
Assume that $(c_i(t))$ is a solution of the Becker-D\"oring equations and $(h_i(t))$ is defined by \eqref{hDefinition}, and let $m \geq 0$. If $h(0) \in X_{3+m}$ then the following is satisfied (strongly) in $X_{1+m}$:
\begin{equation} \label{Eqn:hStrongSolution}
\frac{d}{dt} h = Lh + h_1 \Gamma h.
\end{equation}
In particular, if $h(0) \in X_{3+m}$ then we have that the following is satisfied in $X_{1+m}$:
\begin{equation}\label{Eqn:DuhamelSatisfied}
h(t) = e^{Lt} h(0) + \int_0^t e^{L(t-s)} h_1(s) \Gamma h(s) \ds,
\end{equation}
where $e^{Lt}$ is the semigroup generated by $L$ on $X_{1+m}$ (see Proposition \ref{Cor:XkSemigroup}).
\end{lemma}

\begin{proof}
Because $h(0) \in X_{3+m}$ by Proposition \ref{Prop:BCP} and Lemma \ref{Lemma:GammaEstimate} we have that $Lh + h_1 \Gamma h$ is bounded in $X_{2+m}$ on any finite interval. Because each $h_i$ is continuous by definition \eqref{Def:Solution}, it must be that $Lh + h_1 \Gamma h$ is measurable in $X_{2+m}$. We claim that in $X_{2+m}$ we have that
\begin{equation} \label{Eqn:IntegratedForm}
h(t) = h(0) + \int_0^t Lh(s) + h_1(s) \Gamma h(s) \ds.
\end{equation}
Indeed, the right hand side of the equation is well-defined, and must match the coordinate-wise integrals from definition \ref{Def:Solution}. This implies that $h(t)$ is locally Lipschitz in $X_{2+m}$. As \eqref{Eqn:IntegratedForm} also holds in $X_{1+m}$ we thus have that $h(t)$ must be differentiable in $X_{1+m}$. This implies \eqref{Eqn:hStrongSolution}.

Again by Proposition \ref{Prop:BCP} we know that $h_1\Gamma h \in L^1((0,T);X_{1+m})$. Corollary 4.2.2 in \cite{PazyBook} then implies \eqref{Eqn:DuhamelSatisfied}.
\end{proof}

In deriving uniform bounds we will need a specialized version of Gronwall's inequality.

\begin{lemma}\label{Lemma:Gronwall}
Let $u(t)$ be a positive, continuous function on $[0,\infty)$. Suppose that $u$ satisfies
\begin{equation} \label{Eqn:uInequality}
u(t) \leq C_2(1+t)^{-r} + \int_0^t C_1(1+t-s)^{-r} u(s) ds.
\end{equation}
Furthermore, suppose that $r>1$ and that $C_1$ is small enough that
\begin{equation} \label{Eqn:GronwallConstants}
C_1\int_0^t (1+t-s)^{-r}(1+s)^{-r} \ds \leq \theta(1+t)^{-r}
\end{equation}
for some $\theta<1$ and for all $t>0$. Then we must have that
\[
u(t) \leq \frac{C_2}{1-\theta} (1+t)^{-r}.
\]
\end{lemma}

\begin{proof}
Let $v(t) = u(t)(1+t)^{r}$. Then we have that
\[
v(t) \leq C_2 + (1+t)^r\int_0^t C_1 (1+t-s)^{-r}(1+s)^{-r} v(s) \ds.
\]
This then readily implies that for any $T>0$,
\[
\|v\|_{C(0,T)} \leq C_2 + \theta \|v\|_{C(0,T)}.
\]
Thus for all $t\geq 0$
\[
v(t) \leq \frac{C_2}{1-\theta},
\]
which establishes the desired result.

%
\end{proof}

\begin{remark} \label{Rem:Gronwall}
We note that for any $r>1$ we can find a $C_1>0$ such that \eqref{Eqn:GronwallConstants} is satisfied. This is because
\begin{align}
\int_0^t (1+s)^{-r} (1+t-s)^{-r} \ds &= 2 \int_0^{t/2} (1+s)^{-r}(1+t-s)^{-r} \ds\\
&\leq 2\left(1+\frac{t}{2}\right)^{-r} \int_0^{t/2} (1+s)^{-r} \ds  \\
&\leq \frac{2^{r+1}}{r-1} (1+t)^{-r}.
\end{align}
Thus if $C_1<(r-1)2^{-(r+1)}$ then we have that \eqref{Eqn:GronwallConstants} is satisfied.
\end{remark}

\begin{remark}
The dependence on the constant $C_1$ is critical in the previous proof. Indeed, if $\int_0^\infty C_1 (1+s)^{-r} ds > 1$ then it is possible to show that for some  $u(t) \equiv c>0$ the inequality \eqref{Eqn:uInequality} is satisfied. Thus decay estimates can only be obtained if $C_1$ is sufficiently small.
\end{remark}

The last tool that we need is a local stability estimate. The proof of this estimate is somewhat involved, and we postpone it until the next section.

\begin{theorem} \label{Lem:LyapunovStability}
Let $(c_i)$ be a solution to the Becker-D\"oring equations (see Definition \ref{Def:Solution}), and let $(h_i)$ be determined by \eqref{hDefinition}. Assume that the model coefficients in \eqref{OriginalFluxDefinition} satisfy \eqref{aLowerBound}-\eqref{boundedByI}. Fix $k > 2$. Then given any $\e>0$ there exists $\delta>0$ such that if $\|h(0)\|_{X_{1+k}}<\delta$ then $\|h(t)\|_{X_{1+k}} < \e$ for all $t \geq 0$.
\end{theorem}

With these tools in hand we now prove the main result of this paper.

\begin{proof}[Proof of Theorem \ref{Thm:AlgebraicDecay}]
By Lemma \ref{Lem:Duhamel} we know that the equation
\begin{equation}
h(t) = e^{Lt} h(0) + \int_0^t e^{L(t-s)} h_1(s) \Gamma h(s) \ds
\end{equation}
is satisfied in $X_{1+m}$, where $e^{Lt}$ is the semigroup generated by $L$. By Corollary \ref{Cor:SemigroupPoly} we can thus estimate
\[
\|h(t)\|_{X_{1+m}} \leq C(1+t)^{-(k-m)}\|h(0)\|_{X_{1+k}} + C_L(k,m)\int_0^t \|h_1(s) \Gamma h(s)\|_{X_{k}} (1+t-s)^{-(k-m-1)} \ds.
\]
By Lemma \ref{Lemma:GammaEstimate} we know that $\Gamma$ is bounded from $X_{k+1}$ to $X_{k}$, and thus
\[
\|h(t)\|_{X_{1+m}} \leq C(1+t)^{-(k-m)}\|h(0)\|_{X_{1+k}} + C_L(k,m)C_\Gamma\int_0^t \|h_1(s) h(s)\|_{X_{1+k}} (1+t-s)^{-(k-m-1)} \ds.
\]
It is then immediate that
\[
\|h(t)\|_{X_{1+m}} \leq C(1+t)^{-(k-m)}\|h(0)\|_{X_{1+k}} + C \sup_\tau \|h(\tau)\|_{X_{1+k}} \int_0^t (1+t-s)^{-(k-m-1)} |h_1(s)| \ds.
\]
We then use a crude bound to obtain
\[
\|h(t)\|_{X_{1+m}} \leq C(1+t)^{-(k-m)}\|h(0)\|_{X_{1+k}} + C \sup_\tau \|h(\tau)\|_{X_{1+k}} \int_0^t (1+t-s)^{-(k-m-1)} \|h(s)\|_{X_{1+m}} \ds.
\]
By Lemma \ref{Lem:LyapunovStability} for any $\e>0$ we can choose $\delta_{k,m}$ small enough to guarantee that
\[
\|h(t)\|_{X_{1+m}} \leq C(1+t)^{-(k-m)}\|h(0)\|_{X_{1+k}} + \e \int_0^t (1+t-s)^{-(k-m-1)} \|h(s)\|_{X_{1+m}} \ds.
\]
As $k > m+2$, by applying Lemma \ref{Lemma:Gronwall} (whose conditions will be satisfied for $\e$ small due to Remark \ref{Rem:Gronwall}), we then find that
\[
\|h(t)\|_{X_{1+m}} \leq C(1+t)^{-(k-m-1)}\|h(0)\|_{X_{1+k}},
\]
which is the desired result.
\end{proof}

\section{Local Stability Bounds in $X_{1+k}$} \label{Sec:NonLinear}

In this section our goal is to prove Theorem \ref{Lem:LyapunovStability}. The general strategy is to derive bounds on the evolution family $U(t,s)$ generated by $F(h_1(t))$ when $h_1$ is small. We first establish bounds in $H$ directly using dissipation estimates. We then establish stability bounds on $U(t,s)$ in $X_{1+k}$ by using the extension principle from Proposition \ref{PropMischlerMouhot}. This then immediately implies Theorem \ref{Lem:LyapunovStability}.

\subsection{Non-linear stability in $H$}

The following lemma gives a local, non-linear stability estimate in the space $H$.

\begin{lemma}\label{EvolutionInH}
Suppose that $g(t) \in C^1(I;\R)$, for $I = [0,T)$ with $T$ possibly infinite. Suppose furthermore that the model coefficients in \eqref{OriginalFluxDefinition} satisfy \eqref{aLowerBound}-\eqref{boundedByI}. Then there exist $\delta_H$ and $\lambda > 0$ such that if $|g(t)| < \delta_H$ then $\set{F(g(t))}_{t\in I}$ generates an evolution family $U_H$ in $H$ on the interval $I$ with bound
\[	
\|U_H(t,s)\|_{\mathcal{L}(H)} \leq e^{-\lambda (t-s)} \qquad\text{ for  } 0\le s\le t<T.
\]
\end{lemma}
%

In order to prove this lemma, we will use the following proposition from Pazy (Corollary 5.4.7 and Theorem 5.4.8 in \cite{PazyBook}).

\begin{proposition}\label{NLCPFixedDomain}
Let $X$ be a Banach space and let $I = [0,T)$, with $T=\infty$ permitted. Suppose that, for any fixed $t\in I$, $C(t)$ is the generator of a semigroup $\set{S_{C(t)}(s)}_{s\ge0}$ which satisfies
\[
    \|S_{C(t)}(s)\|_{{\mathcal L}(X)} \leq e^{-\lambda s} \quad\text{for all $s\ge0$},
\] 
where $\lambda$ is independent of $t$. Also suppose that $\dom(C(t)) \equiv D$ is independent of $t$ and that for all $x \in D$ we have that $C(t)x$ is $C^1$ in $X$. Then the family of operators $\set{C(t)}_{t\in I}$ generates an evolution family $U$ on $X$ which satisfies
\begin{equation}
\|U(t,s)\|_{{\mathcal L}(X)} \leq e^{-\lambda(t-s)}
    \qquad\text{ for } 0\le s\le t<T. 
\end{equation}
Furthermore for $x_0 \in D$ we have that $x(t) := U(t,0)x_0$ is the unique solution of the non-autonomous Cauchy problem
\[
\frac{d}{dt} x(t) = C(t) x(t), \qquad x(0) = x_0.
\]
\end{proposition}

Given fixed $N$, we define $T$ to be a diagonal operator given by
\begin{equation}
(Th)_i = -\sigma_i h_i\,, \qquad \sigma_i := Q_1 a_i + b_i\,,\label{Def:SigmaI}
\end{equation}
define $S$ to be the operator
\begin{equation}
(Sh)_i := b_ih_{i-1}\bbone_{\{i>N+1\}} + a_iQ_1 h_{i+1} \bbone_{\{i>N\}}. 
\end{equation}
and $K:= L-T-S$.   In the proof of Lemma \ref{EvolutionInH} we will use the following facts (see Proposition 2.10 and Corollary 2.11 in \cite{CanizoLods}).

\begin{proposition}\label{Prop:CanizoFacts}
Assuming \eqref{aLowerBound}-\eqref{boundedByI}, the operator $L$ given by \eqref{Def:L} satisfies the following properties:
\begin{enumerate}
\item $L$ is self-adjoint in $H$, with $\dom(L) = \dom(T)=\ell^2(Q_i\sigma_i)$.
\item For some $\lambda_c >0$ we have that 
$\langle h,Lh \rangle_H \leq -\lambda_c\|h\|_H^2$ for all $h\in\dom(L)$.
\item $L=T+S+K$,  $K$ is compact on $H$, and for $N$ large enough, 
$S$ is symmetric and satisfies $\|Sh\|_H \leq \theta \|Th\|_H$ for all $h\in\dom(T)$, where $\theta < 1$.
\end{enumerate}

\end{proposition}

We now prove Lemma \ref{EvolutionInH}.

\begin{proof}[Proof of Lemma \ref{EvolutionInH}]
We first claim that the following spectral gap estimate holds as long as $g$ is sufficiently small: For some $\lambda_H>0$,
\begin{equation} \label{L2NLGapEstimate}
\langle F(g) h, h \rangle_H \leq -\lambda_H \|h\|_H^2 \qquad \mbox{ for all } h \in \dom(L).
\end{equation}
To prove this inequality, we recall \eqref{Def:F} and use Proposition \ref{Prop:CanizoFacts} to estimate
\begin{align}
\langle F(g) h, h \rangle_H &= \langle (1-\e) L h, h \rangle_H  + \e \langle K h , h \rangle_H +  \langle (g\Gamma + \e (T+S)) h, h \rangle_H \\
&\leq -(1-\e) \lambda_c \|h\|_H^2 + \e \|K\|_{L(H,H)} \|h\|_H^2 + \langle (g\Gamma + \e (T+S)) h, h \rangle_H.
\end{align}
We  select $\e$ small enough that $\frac{(1-\e)\lambda_c}{2} > \e \|K\|_{L(H,H)}$. As $S$ is $T$-bounded with $T$-bound $\theta<1$ we have that $S$ is relatively bounded (with relative bound smaller than one) by $\frac{1+\theta}{2} T$. Because $S$ is symmetric, this then implies (see \cite{KatoBook}, p.292, Theorem 4.12) that
\[
\langle \left(  S + \left(\frac{1+\theta}{2}\right)T\right)h,h \rangle_H \leq 0.
\]
Thus we can estimate
\begin{align*}
\langle F(g) h, h \rangle_H &\leq  - \frac{(1-\e)\lambda_c}{2} \|h\|_H^2 + \langle \left(\e \frac{1-\theta}{2} T + g \Gamma\right)h,h\rangle_H \\
&= - \frac{(1-\e)\lambda_c}{2} \|h\|_H^2 + \sum_{i=1}^\infty Q_i \left(-\e \frac{1-\theta}{2}\sigma_i h_i^2 + Q_1a_i g h_i(h_{i+1} - h_i-h_1)\right)\\
&\leq - \frac{(1-\e)\lambda_c}{2} \|h\|_H^2 + \sum_{i=1}^\infty Q_i \left(-\e \frac{1-\theta}{2}\sigma_i h_i^2 +\frac{|Q_1 g|}{2} a_i(4h_i^2 + h_{i+1}^2 + h_1^2)\right)\\
&\leq - \frac{(1-\e)\lambda_c}{2} \|h\|_H^2 + \sum_{i=1}^\infty Q_i \left(-\e \frac{1-\theta}{2}\sigma_i+a_i C|Q_1g|\right) h_i^2,
\end{align*}
where we have used the assumptions \eqref{aLimit} and \eqref{QLimit} and the fact that $\sum_{i=1}^\infty Q_i a_i $ is finite. By \eqref{Def:SigmaI} there exists a $\delta_H>0$ so that if $| g|<\delta_H$ then $(a_i C|Q_1 g|-\e \frac{1-\theta}{2}\sigma_i) < 0$. Thus if $|g|<\delta_H$ we deduce that
\[
\langle F(g) h, h \rangle_H \leq - \frac{(1-\e)\lambda_c}{2} \|h\|_H^2 =: -\lambda_H \|h\|_H^2,
\]
which proves the claim.

We observe, from the previous estimates, that indeed $\|\Gamma h\|_H \leq C\|Th\|_h$. This implies that $S + g\Gamma$ is relatively bounded by $T$ with relative bound strictly less than one, as long as $|g| < \delta_H$, where perhaps we have made $\delta_H$ smaller. As $T$ is self-adjoint, by Theorem 1.3.2 in \cite{HenryBook} we have that  $F(g)$ generates an analytic semigroup in $H$. Furthermore, by the relative bound it is clear that $\dom(F(g))  =\dom(T)= \ell^2(Q_i \sigma_i)$.

Now, as $g(t)$ is $C^1$ it is clear that for $v \in D$ we have that $F(g(t))) v$ is in $C^1(I;H)$. We then directly apply Proposition \ref{NLCPFixedDomain} to obtain the desired result.
\end{proof}

\subsection{Non-linear stability in $X_{1+k}$}

The main goal of this subsection is to prove the following lemma.

\begin{lemma}\label{EvolutionInXk}
Suppose that $g(t) \in C^1(I;\R)$, for $I = [0,T)$ with $T$ possibly infinite. Suppose furthermore that the model coefficients in \eqref{OriginalFluxDefinition} satisfy \eqref{aLowerBound}-\eqref{boundedByI} and that $k > 0$. Then there exists a $\delta_k$ such that if $|g(t)| < \delta_k$ then $\set{F(g(t))}_{t\in I}$ generates an evolution family $U_{X_{1+k}}(t,s)$ in $X_{1+k}$ on the interval $I$ with bound
\[
\|U_{X_k}(t,s)\|_{\mathcal{L}(X_{1+k})} \leq M_k,
\]
where $M_k$ is independent of $s,t$ and the particular choice of $g$.
\end{lemma}

To prove this lemma we will use Proposition \ref{PropMischlerMouhot}, in conjunction with the stability in $H$ established in the previous subsection. These techniques should also be applicable in the spaces $X_\eta$, but for the sake of clarity we do not pursue the analysis here.

To begin, we define the operator $A(g)$ in weak form by
\begin{equation} \label{def:AgOperator}
\begin{aligned}
\sum_{i=1}^\infty Q_i (A(g)h)_i \phi_i := &\sum_{i = \tN}^\infty Q_iQ_1a_i (h_i - h_{i+1} + gh_i)(\phi_{i+1}-\phi_i-\phi_1) \\
&  - Q_{\tN-1} Q_1a_{\tN-1} h_{\tN}(\phi_{\tN} - \phi_{\tN-1} - \phi_1) \,,
\end{aligned}
\end{equation}
where $\tN$ is a constant, greater than $N_z+1$, to be determined. We then define $B(g) := F(g) - A(g)$.

The next proposition establishes the dissipativity of $A(g)$.

\begin{proposition} \label{Prop:L1Dissipation}
Under the assumptions of Lemma \ref{EvolutionInXk}, and if $\tN$ in \eqref{def:AgOperator} is chosen large enough, then there exists a constant $\delta_k$ so that if $|g| < \delta_k$ then
\begin{equation}\label{Eqn:DissipationXk}
\langle \sgn{h},A(g) h \rangle_{X_{1+k}^*,X_{1+k}} \leq 0 \quad\mbox{for all $h\in X_{2+k}$.}
\end{equation}
\end{proposition}

\begin{proof} With $w_i = i^{1+k}$ and using $\phi_i=w_i\sgn{h_i}$ in \eqref{def:AgOperator}, we compute, as in the proof of Proposition~\ref{Prop:ADissipative},
\begin{align}
&\langle \sgn{h}, A(g) h \rangle_{X_{1+k}^*,X_{1+k}} 
\\&\quad = 
\sum_{i = \tN}^\infty Q_ih_i\Bigl(
Q_1a_iw_{i+1}(\sgn{h_{i+1}} -\sgn{h_i})
+b_i w_{i-1} (\sgn{h_{i-1}}-\sgn{h_i})
\Bigr)
\\&\qquad + 
\sum_{i = \tN}^\infty Q_i|h_i| (a_iQ_1(w_{i+1} - w_i) + b_i(w_{i-1} - w_{i}))
\\&\qquad
+ \sgn{h_1} \sum_{i = \tN}^\infty Q_ih_i(b_i-Q_1a_i) 
\\&\qquad
 + g\sum_{i = \tN}^\infty Q_i h_iQ_1a_i (w_{i+1}\sgn{h_{i+1}} - w_i \sgn{h_i}-\sgn{h_1}) 
\\&\quad =: E_1 + E_2 + E_3 + E_4 .
\end{align}
First we estimate $E_2$, written as
\begin{align*}
E_2
&= \sum_{i=\tN}^\infty Q_i|h_i|(w_{i+1} - w_i)\left(a_i Q_1 -b_i \frac{w_i-w_{i-1}}{w_{i+1}-w_i}\right).
\end{align*}
By choosing $\tN$ sufficiently large we can make the ratio $\frac{w_i-w_{i-1}}{w_{i+1}-w_i}$ arbitrarily close to $1$. Thus we apply \eqref{asymptoticallyStrongFragmentation} to find that
\[
E_2 \leq - C \sum_{i=\tN}^\infty Q_i |h_i| a_i (w_{i+1} - w_i).
\]
We next calculate
\[
E_3 
\leq \sum_{i = \tN}^\infty Q_i  |h_i|(b_i+Q_1a_i).
\]
Recalling \eqref{tildeQLimit}, and using that $w_{i+1}-w_i\to\infty$ since {$k>0$}, we thus have, for $\tN$ sufficiently large,
\[
E_2 + E_3 \leq - C \sum_{i=\tN}^\infty Q_i |h_i|a_i (w_{i+1} - w_i) .
\]
Because $h_i(\sgn{h_{i\pm1}}-\sgn{h_i})\le0$, we infer $E_1 \leq 0$. 
  Thus, in the case $g \geq 0$ we estimate
\begin{align}
E_1+E_4 \leq E_4 
&\leq |g|
\sum_{i = \tN}^\infty Q_i Q_1a_i|h_i| (w_{i+1}- w_i +1) 
\\&\leq C| g| \sum_{i=\tN}^\infty Q_i |h_i|a_i (w_{i+1} - w_i).\label{Eqn:gPositive}
\end{align}
For $g < 0$  we find that
\begin{align}
E_1+ E_4 
&\leq \ %
\sum_{\mathclap{\substack{i \geq \tN\\ \sgn{h_i} \neq \sgn{h_{i+1}}}}}
Q_i|h_i|Q_1a_i(-2w_{i+1} - g(w_{i+1} + w_i))
\\&\quad +
|g| \sum_{i = \tN}^\infty Q_i|h_i| Q_1a_i
\label{Eqn:gNegative}
\end{align}
When $|g| < 1$ we have that the first term in \eqref{Eqn:gNegative} is negative. This then readily implies that for $\tN$ sufficiently large and for $|g|$ sufficiently small we have that
\[
\langle \sgn{h}, A(g) h \rangle_{X_{1+k}^*,X_{1+k}} \leq - C \sum_{i=\tN}^\infty Q_i (w_{i+1} - w_i) a_i |h_i| \leq 0,
\]
which completes the proof.
\end{proof}

The next step is to prove that $\set{A(g(t))}$ indeed generates an evolution family.

\begin{lemma}\label{lem:AgEvolutionFamily}
Suppose that the assumptions of Lemma \ref{EvolutionInXk} are satisfied. Suppose furthermore that
\begin{equation}\label{Eqn:gBound}
|g(t)|<\min\{\delta_k,\delta_{k+1},\delta_H\},
\end{equation}
where $\delta_k$ is given in Proposition \ref{Prop:L1Dissipation} and $\delta_H$ in Lemma \ref{EvolutionInH}. Then for $\tN$ chosen as in Proposition \ref{Prop:L1Dissipation}, the family $\set{A(g(t))}_{t\in I}$ generates an evolution family $V_{X_{1+k}}$ on the interval $I = [0,T)$ in the space $X_{1+k}$, which for $0\leq s \leq t <T$ satisfies
\[
\|V_{X_{1+k}}(t,s)\|_{\mathcal{L}(X_{1+k})} \leq 1.
\]
\end{lemma}

To prove this lemma, we use the following proposition, which is a direct application of Theorem 5.3.1 in \cite{PazyBook}.

\begin{proposition}\label{PazyNLCP2}
Let $I=[0,T)$, with $T=\infty$ permitted, and suppose that a family of linear operators $\set{C(t)}_{t\in I}$ satisfies the following for all $t \in I$.
\begin{enumerate}
\item $C(t)$ generates a contraction semigroup on $X_{1+k}$.
\item $C(t)$ generates a contraction semigroup on $X_{2+k}$.
\item $C(t)$ is a bounded operator from $X_{2+k}$ to $X_{1+k}$, and the map $t\mapsto C(t)$ is continuous from $I$ to ${\cal L}(X_{2+k},X_{1+k})$.
\end{enumerate}
Then $\set{C(t)}_{t\in I}$ generates an evolution family $V_{X_{1+k}}$ satisfying $\|V_{X_{1+k}}(t,s)\|_{\mathcal{L}(X_{1+k})} \leq 1$.
\end{proposition}

We now prove Lemma \ref{lem:AgEvolutionFamily}.

\begin{proof}[Proof of Lemma \ref{lem:AgEvolutionFamily}]
We claim that $\set{A(g(t))}_{t\in I}$ satisfies the assumptions of Proposition \ref{PazyNLCP2}. By \eqref{Eqn:gBound} and Proposition \ref{Prop:L1Dissipation} we have that $A(g(t))$ is dissipative on $X_{1+k}$ and $X_{2+k}$. For fixed $t\in I$, by \eqref{Eqn:gBound}, $F(g(t))$ generates a semigroup on $H$ (as established in the proof of Lemma \ref{EvolutionInH}). As $B(g(t))$ is a bounded operator on $H$, it then must be that $A(g(t))$ generates a semigroup on $H$. This then implies that for some large, positive real $\lambda$ we must have that
the range of  $A(g(t))-\lambda$ contains  $H$. Thus the range of $A(g(t)) - \lambda$ is dense in $X_{1+k}$ and $X_{2+k}$. As in the proof of Lemma \ref{Lem:AContraction}, this implies that $A(g(t))$ generates a semigroup on $X_{1+k}$ and $X_{2+k}$, and thus the first two assumptions are satisfied.

Next, as $g(t)$ is $C^1$ and by \eqref{boundedByI}, the third assumption is necessarily satisfied. Thus we may apply  Proposition \ref{PazyNLCP2}, which proves the lemma.
\end{proof}

We remark that the previous lemma is independent of the choice of $g(t)$. The next result follows from a computation as in Lemma \ref{Lemma:BRegularizing}, and we omit the proof.

\begin{lemma}\label{Lem:BgRegularizing}
Under the assumptions of Lemma \ref{lem:AgEvolutionFamily}, the operator $B(g(t))$ is uniformly bounded from $X_1$ to $H$, with a bound that depends only on $\delta_{k}$, and not on $g$ or $t$.
\end{lemma}

We are now in a position to prove Lemma \ref{EvolutionInXk}.

\begin{proof}[Proof of Lemma \ref{EvolutionInXk}]
In light of Lemmas \ref{lem:AgEvolutionFamily} and \ref{Lem:BgRegularizing} this follows  from Proposition \ref{PropMischlerMouhot}.
\end{proof}

\begin{remark}
We note that the bound $M_k$ is not dependent on the particular function $g(t)$, and only on its bound $\delta_{k}$. This is because of the independence on $g(t)$ in the bounds obtained in lemmas \ref{lem:AgEvolutionFamily} and \ref{Lem:BgRegularizing}.
\end{remark}

The following is a direct application of Theorem 5.4.2 in \cite{PazyBook}.

\begin{proposition}\label{Prop:ClassicalSol}
Suppose that, for some $k \geq 0$, $\set{C(t)}_{t\in I}$ generates an evolution family $U$ in $X_{1+k}$ on the interval $I=[0,T)$. Furthermore, suppose that for some $h \in C(I;X_{2+k}) \bigcap C^1(I;X_{1+k})$ we have that
\[
\frac{d}{dt} h = C(t) h(t)
\]
is satisfied in $X_{1+k}$. Then it must be that $U(t,0)h(0) = h(t)$.
\end{proposition}

Now we finally give the proof of Theorem \ref{Lem:LyapunovStability}.

\begin{proof}[Proof of Theorem \ref{Lem:LyapunovStability}]
 Let $M_k$ be the uniform bound in the space $X_k$ given in Lemma \ref{EvolutionInXk}. Set 
 \[
 \delta = \frac{Q_1\min\{\delta_{k-2}, \delta_{k-1}, \delta_k, \delta_{k+1}, \delta_H,\e Q_1^{-1} \}}{2 M_k}.
 \]
 Now, let $(h_i)$ correspond to a solution of the Becker-D\"oring equations, with $\|h(0)\|_{X_{1+k}} < \delta$. By Lemma \ref{Lem:Duhamel} and as $k > 2$ we know that $h_1$ is $C^1$. By Lemma \ref{EvolutionInXk} we thus know that $\set{F(h_1(t))}_{t\in I}$ generates an evolution family $U$ on $X_{1+(k-2)}$ and $X_{1+k}$ on the (non-empty) interval $I$ such that $|h_1(t)| \leq \min\{\delta_{k-2}, \delta_{k-1}, \delta_k, \delta_{k+1}, \delta_H\}$. As $k > 2$, by Lemma \ref{Lem:Duhamel} we know that the conditions of Proposition \ref{Prop:ClassicalSol} are satisfied in $X_{1+(k-2)}$, and thus $U(t,0)h(0) = h(t)$ for all $t \in I$. 
 
The uniform bounds from Lemma \ref{EvolutionInXk} then imply that $\|h(t)\|_{X_{1+k}} \leq M_i \|h(0)\|_{X_{1+k}}$ on the interval $I$. Our choice of $\delta$ immediately implies that $I = [0,\infty)$ and that $\|h(t)\|_{X_k} \leq \e/2$, which completes the proof.
\end{proof}

\section*{Acknowledgments}
This material is based upon work supported by the National
Science Foundation under grants DMS 1211161 and DMS 1515400,
and also supported by the Center for Nonlinear Analysis (CNA)
under National Science Foundation PIRE Grant no.\ OISE-0967140.
RLP also acknowledges support from F.C.T. (Portugal) grant UTA CMU/MAT/0007/2009.

\bibliography{BDquasilinearRefs}

\begin{thebibliography}{10}

\bibitem{BallCarrConvergence}
J.~M. Ball and J.~Carr.
\newblock Asymptotic behaviour of solutions to the {B}ecker-{D}\"oring
  equations for arbitrary initial data.
\newblock {\em Proc. Roy. Soc. Edinburgh Sect. A}, 108(1-2):109--116, 1988.

\bibitem{BallCarrPenrose}
J.~M. Ball, J.~Carr, and O.~Penrose.
\newblock The {B}ecker-{D}\"oring cluster equations: basic properties and
  asymptotic behaviour of solutions.
\newblock {\em Comm. Math. Phys.}, 104(4):657--692, 1986.

\bibitem{BeckerDoring}
R.~Becker and W.~D{\"o}ring.
\newblock Kinetische behandlung der keimbildung in {\"u}bers{\"a}ttigten
  d{\"a}mpfen.
\newblock {\em Annalen der Physik}, 416(8):719--752, 1935.

\bibitem{BerghLofstrom}
J{\"o}ran Bergh and J{\"o}rgen L{\"o}fstr{\"o}m.
\newblock {\em {Interpolation Spaces: An Introduction}}, volume 223.
\newblock Springer, 1976.

\bibitem{Burton}
J.J. Burton.
\newblock Nucleation theory.
\newblock In Bruce~J. Berne, editor, {\em Statistical Mechanics}, volume~5 of
  {\em Modern Theoretical Chemistry}, pages 195--234. Springer US, 1977.

\bibitem{CanizoLods}
Jos{\'e}~A. Ca{\~n}izo and Bertrand Lods.
\newblock Exponential convergence to equilibrium for subcritical solutions of
  the {B}ecker-{D}\"oring equations.
\newblock {\em J. Differential Equations}, 255(5):905--950, 2013.

\bibitem{EngelNagel}
Klaus-Jochen Engel and Rainer Nagel.
\newblock {\em One-parameter semigroups for linear evolution equations}, volume
  194.
\newblock Springer Science \& Business Media, 2000.

\bibitem{Engler}
Hans Engler.
\newblock Asymptotic stability of traveling wave solutions for perturbations
  with algebraic decay.
\newblock {\em J. Differential Equations}, 185(1):348--369, 2002.

\bibitem{MouhotMischler}
M.~P. {Gualdani}, S.~{Mischler}, and C.~{Mouhot}.
\newblock {Factorization for non-symmetric operators and exponential
  H-theorem}.
\newblock arXiv:1006.5523.

\bibitem{HenryBook}
Daniel Henry.
\newblock {\em Geometric Theory of Semilinear Parabolic Equations}.
\newblock Lecture Notes in Mathematics. Springer Berlin Heidelberg, 1981.

\bibitem{JabinNiethammer}
Pierre-Emmanuel Jabin and Barbara Niethammer.
\newblock On the rate of convergence to equilibrium in the {B}ecker-{D}\"oring
  equations.
\newblock {\em J. Differential Equations}, 191(2):518--543, 2003.

\bibitem{KatoBook}
Tosio Kato.
\newblock {\em Perturbation theory for linear operators}.
\newblock Classics in Mathematics. Springer-Verlag, Berlin, 1995.
\newblock Reprint of the 1980 edition.

\bibitem{LaurencotMischlerUniqueness}
Philippe Lauren{\c{c}}ot and St{\'e}phane Mischler.
\newblock From the {B}ecker-{D}\"oring to the {L}ifshitz-{S}lyozov-{W}agner
  equations.
\newblock {\em J. Statist. Phys.}, 106(5-6):957--991, 2002.

\bibitem{PazyBook}
A.~Pazy.
\newblock {\em Semigroups of linear operators and applications to partial
  differential equations}, volume~44 of {\em Applied Mathematical Sciences}.
\newblock Springer-Verlag, New York, 1983.

\bibitem{PenroseLebowitz}
O~Penrose and Joel~L Lebowitz.
\newblock Towards a rigorous molecular theory of metastability.
\newblock {\em Fluctuation Phenomena}, 7:293--340, 1987.

\bibitem{Penrose89}
Oliver Penrose.
\newblock Metastable states for the {B}ecker-{D}{\"o}ring cluster equations.
\newblock {\em Communications in Mathematical Physics}, 124(4):515--541, 1989.

\bibitem{SlemrodEquilibrium}
M.~Slemrod.
\newblock Trend to equilibrium in the {B}ecker-{D}\"oring cluster equations.
\newblock {\em Nonlinearity}, 2(3):429--443, 1989.

\bibitem{Slemrod}
Marshall Slemrod.
\newblock The {B}ecker-{D}{\"o}ring equations.
\newblock In {\em Modeling in applied sciences}, pages 149--171. Springer,
  2000.

\end{thebibliography}
\bibliographystyle{siam} 

\end{document}